\documentclass[letterpaper,11pt]{article}
\usepackage{graphicx} % Required for inserting images
\usepackage[]{amsmath,amssymb,amsfonts,latexsym,amsthm,enumerate,fullpage,xcolor,bbm,mathtools}
\usepackage[pagebackref,colorlinks,citecolor=blue,urlcolor=magenta,linkcolor=magenta,,bookmarks=true]{hyperref}
\usepackage[nameinlink, noabbrev, capitalize]{cleveref}
\usepackage{amsmath}
\usepackage[noadjust]{cite}

\title{Most Juntas Saturate the Hardcore Lemma}
\author{Vinayak M. Kumar\thanks{\href{mailto:vmkumar@utexas.edu}{\texttt{vmkumar@utexas.edu}}. Department of Computer Science, The University of Texas at Austin. Supported in part by NSF Grant CCF-2312573, a Simons Investigator Award (\#409864, David Zuckerman), a Jane Street Graduate Research Fellowship, and a UT Austin Dean's Prestigious Fellowship Supplement.}}
\date{}
\renewcommand{\backref}[1]{}

\renewcommand{\backrefalt}[4]{%
\ifcase #1 %
\or
[p.\ #2]%
\else
[pp.\ #2]%
\fi}

\def\colorful{1}

\ifnum\colorful=1

\fi
\ifnum\colorful=0

\fi

\newtheorem{theorem}{Theorem}
\newtheorem{lemma}{Lemma}

\newtheorem{conjecture}{Conjecture}

\newtheorem{definition}{Definition}

\newtheorem{remark}{Remark}

\usepackage{comment}
\usepackage{nicefrac}
\crefname{prop}{Proposition}{Propositions}
\crefname{ineq}{inequality}{inequalities}
\creflabelformat{ineq}{#2(#1)#3}

\crefname{THM}{Theorem}{Theorems}

\newcommand{\E}{\mathbb{E}}
\newcommand{\N}{\mathbb{N}}

\newcommand{\F}{\mathbb{F}}

\newcommand{\R}{\mathbb{R}}

\newcommand{\eps}{\varepsilon}

\newcommand{\enc}{{\textnormal{\textsf{Enc}}}}

\newcommand{\cC}{\mathcal{C}}

\begin{document}
    \maketitle

\begin{abstract}
    Consider a function that is mildly hard for size-$s$ circuits. For sufficiently large $s$, Impagliazzo’s hardcore lemma guarantees a constant-density subset of inputs on which the same function is extremely hard for circuits of size $s' \ll s$. Blanc, Hayderi, Koch, and Tan [FOCS 2024] recently showed that the degradation from $s$ to $s'$ in this lemma is quantitatively tight in certain parameter regimes. 
    We give a simpler and more general proof of this result in almost all parameter regimes of interest by showing that a random junta witnesses the tightness of the hardcore lemma with high probability.

    %arxiv
    %Consider a function that is mildly hard for size-$s$ circuits. For sufficiently large $s$, Impagliazzo's hardcore lemma guarantees a constant-density subset of inputs on which the same function is extremely hard for circuits of size $s'<\!\!<s$. Blanc, Hayderi, Koch, and Tan [FOCS 2024] recently showed that the degradation from $s$ to $s'$ in this lemma is quantitatively tight in certain parameter regimes. We give a simpler and more general proof of this result in almost all parameter regimes of interest by showing that a random junta witnesses the tightness of the hardcore lemma with high probability.

    %Consider a function that is mildly hard for size-$s$ circuits. For sufficiently large $s$, Impagliazzo's hardcore lemma provides a subset of inputs on which the same function is extremely hard for circuits of size $s'\ll s$. Blanc, Hayderi, Koch, and Tan [FOCS, 2024] proved that the size degradation from $s$ to $s'$ present in this lemma is quantitatively tight. We give a much simpler proof of this fact by showing that the hardcore lemma is tight for a random junta with high probability. 
    
    %Blanc, Hayderi, Koch, and Tan [FOCS, 2024] proved that Impagliazzo's Hardcore Lemma is quantitatively tight. We give a much simpler proof of this fact by showing that the lemma is tight for a random junta with high probability.
\end{abstract}

\section{Introduction}

Let $f:\{0,1\}^n \to \{0,1\}$ be a Boolean function such that every circuit of size $s$ errs on at least a $\delta$-fraction of inputs. 
How can we amplify the hardness of this function? One approach is to restrict the domain: given a fixed size-$s$ circuit, we select a subset of inputs of density at least $2\delta$ in which half the points come from the error region and half are correct. 
Such a set forms a \emph{hardcore set}, because on this region the circuit cannot do better than random guessing.
Is it possible that there exists a \emph{single} subset of density $2\delta$ that is simultaneously hard for all size-$s$ circuits?
Impagliazzo's hardcore lemma establishes the existence of a $\Omega(\delta)$-density hardcore set for \emph{all} circuits of size $s'\ll s$. The version of this lemma with the smallest size degradation from $s$ to $s'$ is the following.

\begin{theorem}[\cite{Impagliazzo1995, KS03, BarakHardcore2009}]
\label{thm:imp-tight}
    Let $f:\{0,1\}^n\to\{0,1\}$ and $\delta, \gamma, n\le s \le  \tfrac{2^n}{n}$. Suppose that for all circuits $C$ of size at most $s$, 
    \[\Pr_{x\sim \{0,1\}^n}[C(x)=f(x)]\le 1-\delta.\] 
    Then there exists a subset $H\subset\{0,1\}^n$ of density $\Omega(\delta)$ such that for all circuits $C$ of size $O\left(\frac{s\gamma^2}{\log(1/\delta)}\right)$, we have 
    \[\Pr_{x\sim H}[C(x) = f(x)]\le \frac{1}2 + \gamma.\]
\end{theorem}

Conceptually, the theorem says that circuit hardness can be explained by a subset of ``hard inputs'' $H$ on which the function looks random to small circuits.\footnote{Holenstein \cite{Holenstein2005} gives a set $H$ of optimal density $2\delta$, but suffers a larger degradation from $s$ to $O\left(\frac{s\gamma^2}n\right)$.} This phenomenon has found applications throughout computer science, including hardness amplification \cite{ODonnell-STOC-2002, Trevisan2003}, pseudorandomness \cite{SudanTrevisanVadhan2001, ChenLyu2021}, cryptography \cite{Holenstein2005}, algorithmic fairness \cite{CasacubertaDworkVadhan2024}, combinatorics \cite{ReingoldTrevisanTulsianiVadhan2008}, and learning theory \cite{BarakHardcore2009,KS03}.

While the hardcore lemma is a remarkable result, a natural question is whether the size degradation $s\to \tfrac{s\gamma^2}{\log(1/\delta)}$ is necessary. This is formalized by the following conjecture. 

\begin{conjecture}
\label{conj:tight}
For any $\delta\in (0,1),\gamma\in (0,\tfrac{1}2)$, $n \in \N$ large enough, and $ \Omega\left(\tfrac{\log(1/\delta)}{\gamma^2}\right)\le s\le  O\left(\tfrac{2^n}{n}\right)$, there exists an $f$ such that 
\begin{itemize}
 \item for all circuits $C$ of size $\le s$, \[\Pr_{x\sim \{0,1\}^n}[C(x) = f(x)]\le 1-\delta.\]
 \item for every subset $H\subset\{0,1\}^n$ of density $\ge \Omega(\delta)$, there exists a circuit $C$ of size $O\left(\tfrac{s\gamma^2}{\log(1/\delta)}\right)$ such that \[\Pr_{x\sim H}[C(x) = f(x)] \ge \frac{1}2 + \gamma.\]
\end{itemize}
\end{conjecture}

Such a degradation was shown to be necessary for a certain class of proofs \cite{LuTsaiWu2011}, but an unconditional result remained elusive, as proving this theorem appears to require constructing an explicit, mildly approximable function that demands large circuits to strongly approximate it. This felt tantamount to proving breakthrough circuit lower bounds. In a recent work, Blanc, Hayderi, Koch, and Tan \cite{BHKT24} evaded this barrier by arguing about a nonexplicit function. In fact, the proof of \Cref{thm:imp-tight} first shows the result for $H\sim\{0,1\}^n$ being the weaker notion of a $\delta$-smooth distribution (i.e., no string has more than $\frac{1}{\delta 2^n}$ probability of being sampled), and then uses the distribution to extract a density-$\delta$ subset. \cite{BHKT24} proves a tightness result over the more general $\delta$-smooth distributions. 

\begin{theorem}[{{\cite[Theorem 2]{BHKT24}}}]
\label{thm:bhkt}
    Let $\delta\in (0,1), \gamma\in \left[\Omega\left(\tfrac{1}{\sqrt{n}}\right),\tfrac{1}2\right)$, $n\in \N$ large enough. For $s\in[ \Omega(\tfrac{1}{\gamma^2}), O(\frac{2^{\gamma^2 n}}{\gamma^4 n})]$, there exists $f:\{0,1\}^n\to\{0,1\}$ such that 
    \begin{itemize}
    \item For all circuits $C$ of size $\le s$, \[\Pr_{x\sim \{0,1\}^n} [C(x) = f(x)]\le 1-\delta.\]
     \item for all $\delta$-smooth distributions $H\sim \{0,1\}^n$, there exists a circuit $C$ of size $O\left(s\gamma^2\right)$ such that \[\Pr_{x\sim H}[C(x)=f(x)]\ge \frac{1}2 + \delta\gamma.\]
    \end{itemize}
\end{theorem}
The above is a reparametrized version of what appears in \cite{BHKT24}, which includes dependencies on $\delta$ from their proof (where $\delta$ was assumed to be constant).\footnote{ \cite[Theorem 2]{BHKT24} gives, for any parameters $s$ and $\gamma$, a tightness result for a function of input length $n(s,\gamma)$. We have reparametrized to the standard convention of fixing the input length to $n$, and examining which $s,\gamma$ are possible (in terms of $n$).} Hence, in the regime $\delta = \Theta(1)$, $\gamma \ge \Omega(1/\sqrt{n}) $ and $s = O(2^{\gamma^2 n}/(\gamma^4 n))$, \cite{BHKT24} shows that the $\gamma^2$-factor decay in size is tight. Structurally, their argument is very analytically involved and is in multiple stages. In what follows, we say $f$ $\gamma$-correlates with $g$ over $H$ if $\Pr_{x\sim H}[f(x) = g(x)]\ge \frac{1}2 + \gamma$, and $f$ $\gamma$-approximates $g$ over $H$ if $\Pr_{x\sim H}[f(x) = g(x)]\ge \gamma$.

\begin{itemize}
    \item Let $k = 1/\gamma^2$. They first prove that for any $\delta$-smooth distribution $H$, the majority on $k$ bits is $\tfrac{\delta}{\sqrt{k}}$-correlated with a $1$-junta over $H$, but no $0.01k$-junta can $\tfrac{1}4$-correlate with it over the uniform distribution.\footnote{A $k$-junta $f:\{0,1\}^n\to \{0,1\}$ is a Boolean function depending only on a subset of $k<n$ input variables. The junta complexity of a function is the smallest $k$ such that the function is a $k$-junta.} 
    \item They then bootstrap this result to show that, for any $\delta$-smooth $H$, the majority of $k$ random functions on $k$ disjoint $\tfrac{n}k$-bit input blocks is $\tfrac{\delta}{\sqrt{k}}$-correlated with a size-$O\left(\frac{2^{n/k}}{n/k}\right)$ circuit over $H$, but requires size $\Omega\left(\frac{k2^{n/k}}{n/k}\right)$ to $\tfrac{1}4$-correlate with (over the uniform distribution). They accomplish this by introducing an analytic relaxation of junta complexity, using Fourier-analytic noise-stability arguments to equate this relaxation to the original measure (up to constant factors), and then using coupling arguments, Fourier-analytic calculations, and subgaussian concentration to analyze the relaxed junta complexity of the random-function ensemble.
\end{itemize}
We note this junta-to-circuit lifting theorem is interesting in its own right, and towards proving this, they establish a novel direct sum theorem. These techniques are also used in \cite{BHKT24} to tightly characterize the sample-complexity overhead of smooth boosting. 

In this note, we give a very short proof that a random junta saturates the hardcore lemma in the regime $\delta = \Omega(1)$ over \emph{arbitrary} distributions.

\begin{theorem}
\label{thm:random-tight}
    Let $\delta\in (0, 0.49), \gamma\in (0,\tfrac{1}2), n\in \N$ large enough, and $n\le  s \le O\left(\tfrac{2^n}n\right)$. If there exists a constant $\eps > 0$ such that $s \ge 1/\gamma^{2 + \eps}$, there exists a function $f$ such that
    \begin{itemize}
     
    \item For all circuits $C$ of size $s$, \[\Pr_{x\sim \{0,1\}^n}[C(x) = f(x)]\le 1-\delta.\]

    \item for all distributions $H$ over $\{0,1\}^n$, there exists a circuit $C$ of size $O(s\gamma^2)$ such that \[\Pr_{x\sim H}[C(x) = f(x)]\ge \frac{1}2+\gamma.\]
    
    \end{itemize}
\end{theorem}

Formally, \Cref{thm:bhkt} and \Cref{thm:random-tight} are incomparable: \Cref{thm:bhkt} holds only for $\gamma \ge \Omega(1/\sqrt{n})$ and $s\in [\Omega(1/\gamma^2), O(2^{\gamma^2 n}/(\gamma^4 n))]$, while \Cref{thm:random-tight} holds for any $\gamma$ and any $s\in [\Omega(1/\gamma^{2+\eps}), O(2^{ n}/ n)]$. While incomparable, we note the former region is extremely restrictive, and does not include the common setting of $\gamma = \tfrac{1}{n}$. Meanwhile, the latter region contain almost all possible $(\gamma, s)$ pairs. The latter region is short of subsuming the former interval only by an arbitrarily small polynomial factor of $\gamma^{-\eps}$. Removing this $\eps$-slack remains open. 

When $s\ge 1/\gamma^{2+\eps}$, \Cref{thm:random-tight} improves on \Cref{thm:bhkt} in two ways:  
\begin{itemize}
    \item \Cref{thm:bhkt} constructs circuits with correlation $\delta \gamma$, while \Cref{thm:random-tight} has correlation $\gamma$ that does not degrade with $\delta$.
    \item \Cref{thm:bhkt} requires $H$ to be $\delta$-smooth, but \Cref{thm:random-tight} makes no assumption about $H$. 
\end{itemize}
These two points allow our theorem to remain meaningful even for $\delta = o(1)$. Interestingly, the second point shows that the $s\to s\gamma^2$ decay is unavoidable \emph{even if we permit $H$ to be of density $o(\delta)$}.
     
In summary, our \Cref{thm:random-tight} is stronger in the regime $s\ge 1/\gamma^{2+\eps}$ or $\gamma = O(1/\sqrt{n})$, but gives inferior bounds when $1/\gamma^2 \le s \le  (1/\gamma)^{2+\eps}$ and $\gamma = \Omega(1/\sqrt{n})$.\footnote{For general $s = \Omega(1/\gamma^2)$, we get approximating circuits of size $O\left(s\gamma^2\cdot  \frac{\log s}{\log(s\gamma^2)}\right)$. See \Cref{rem:explicit-circuit-size}.} We also note that \Cref{conj:tight} remains open. In particular, it would be interesting to pin down the optimal dependence of the circuit-size decay on $\delta$.

Our main technical lemma is the strengthening of a beautiful result of Andreev, Clementi, and Rolim \cite{acr95}, which shows that arbitrary Boolean functions can be approximated by small-size circuits.
%\gnote{It looks like the result in \cite{acr95}'s abstract already handles partial functions, so would work over any \emph{set} $H$. }

\begin{theorem}
\label{thm:approx}
    For an arbitrary function $f:\{0,1\}^n\to \{0,1\}$, $\tfrac{27}{2^{n/2}} < \gamma < \tfrac{1}2$, and any distribution $H$ over $\{0,1\}^n$, there exists a circuit $C$ of size $O\left(\frac{\gamma^2 2^n}{\log(\gamma^2 2^n)} + n\right)$ such that \[\Pr_{x\sim H}[C(x) = f(x)]\ge \frac{1}2+\gamma.\]
\end{theorem}

This result was proven by Andreev, Clementi, and Rolim \cite{acr95} in the case where $H$ is uniform over $\{0,1\}^n$. A short probabilistic argument proves that the circuit size cannot be improved \cite[Theorem 4.1]{acr95}.

\subsection{Proof Overview}

Assume $s = 2^k/k$ for some integer $k$, and let $H\sim \{0,1\}^n$ be an arbitrary distribution. Intuitively, our proof of \Cref{thm:random-tight} will first use the classic Shannon argument to show that, with high probability, a random function on the first $k$ bits cannot be $(1-\delta)$-approximated by circuits of size $s$. Letting $H'$ be the induced distribution of $H$ on the first $k$ bits, we can use \Cref{thm:approx} to see that there exist circuits of size $O\left(\frac{\gamma^2 2^k}{\log(\gamma^2 2^k)} + k\right) = O(s\gamma^2)$ that $\gamma$-correlate with $f$ over $H'$ whenever $s\ge 1/\gamma^{2 + \eps}$. The combination of both of these claims implies \Cref{thm:random-tight}. 

In the main body of the paper, we will actually show a slightly weaker version of \Cref{thm:approx} with circuit size $O\left(\frac{\gamma^2 2^n}{\log(\gamma^2 2^n)} + n^2\right)$. The proof of this claim is drastically simpler than that of \Cref{thm:approx}, but still recovers \Cref{thm:random-tight}.

Although the weaker version of \Cref{thm:approx} suffices, it is our impression that the result of \cite{acr95} is not as well known to the community as it should be. It is a complete resolution to a very natural question in circuit complexity (it is the ``approximate version'' of Lupanov's theorem \cite{lupanov}), and the ideas behind the construction are quite useful (e.g., they appear to have been rediscovered in the construction of covering codes of Rabani and Shpilka \cite{RabaniShpilka2010}). This is potentially due to the paper being quite technical and terse, as well as evading search engines. For this reason, we hope to bring attention to this result by giving an exposition and simpler, self-contained proof of \Cref{thm:approx} in the appendix, assuming a basic consequence of the fourth-moment method \cite{berger} and the existence of asymptotically good codes encodable by linear-size circuits \cite{Spielman1996}. We now provide a proof overview below.

\subsubsection{Overview of \Cref{thm:approx}}
\label{subsubsec:overview}

For ease of exposition, we give an overview of \Cref{thm:approx} when $H$ is fixed to be uniform over $\{0,1\}^n$. Extending the given arguments to arbitrary distributions $H$ just requires a couple of extra modifications. This overview is morally the same as one provided by Trevisan \cite{Trevisan2009Blog}, but perhaps with a slightly different point of view in the latter half.

The initial observation is that a random function can actually be efficiently approximated. By standard anticoncentration results, the bias of a random function will be at least $\Omega(2^{-n/2})$ with constant probability, in which case either the constant $0$ or $1$ function will give a $\left(\frac{1}2 + \Omega(2^{-n/2})\right)$-approximation of $f$ (we will henceforth refer to this as \emph{square-root anticoncentration}). For a size–approximation trade-off, we can split the truth table of $f$ into $2^k$ subcubes according to the first $k$ bits of the input, and then approximate each subcube by its majority bit. This is a function depending only on the first $k$ bits of input and can therefore be implemented by a size-$O(2^k/k)$ circuit. Since each subcube will have $2^{n-k}$ bits, it follows that for a random function, the majority bit will give a $\left(\tfrac{1}2 + \Omega(\sqrt{2^{-(n-k)}})\right)$-approximation of the subcube with constant probability. Hence, with high probability, a constant fraction of subcubes will be $\left(\frac{1}2 + \Omega(\sqrt{2^{-(n-k)}})\right)$-approximated (say, by Chernoff), and thus this circuit will have overall approximation $\tfrac{1}2 + \Omega\left(\sqrt{2^{-(n-k)}}\right)$ with $f$. Setting $k = \log(\gamma^2 2^n)$ gives the result for a random function.
%\footnote{One might think proving \Cref{thm:approx} for a random function suffices to show \Cref{thm:random-tight}. This argument shows that for a fixed distribution $H$, a random function will have a small-size approximator, but will not say that a random function will have small-size approximators for \emph{all} $H$. }

Does this argument work for an arbitrary (rather than a random) $f$? Clearly not: one can pick any $f$ that is unbiased on each of these $2^k$ subcubes (e.g., the parity function), and the constructed circuit will have no correlation with $f$. 

What if we could ``reduce to the random case'' by artificially adding random noise to the truth table of $f$, and then approximating this noisy function with a circuit on the first $k$ bits? To be more precise, say we had a distribution $\cC$ over size-$s$ circuits such that the truth table of $C\sim \cC$ was a random string. Then we know $f\oplus \cC$ will be a random function, and consequently can be $\left(\tfrac{1}2 + \gamma\right)$-approximated by a function $g$ on the first $k$ bits with high probability. We can then fix such a $C\in \cC$, and deduce that $C\oplus g$ is a good approximator for $f$ with size $O(2^k/k + s)$. 

Unfortunately, a fully random truth table can only be generated by maximally sized circuits. However, we could hope to use a \emph{pseudorandom} string instead. The only property used about the randomness of $f$ was that its truth table had square-root anticoncentration on subcubes. It turns out $4$-wise uniform strings have square-root anticoncentration with constant probability \cite{berger}, motivating us to look at this primitive. Implementing the usual 4-wise uniform-generator construction naively in a circuit immediately gives a distribution $\cC$ over circuits of size $O(n^2)$ such that the truth table of $C\sim\cC$ is a $4$-wise independent string. With more effort, one can get a distribution over $O(n)$-size circuits, which is optimal. By the fourth-moment method \cite{berger}, we can argue that the average number of subcubes of $f\oplus \cC$ with square-root anticoncentration is at least a constant proportion. Fixing $C\in\cC$ that achieves this average, it follows that there is a function $g$ on the first $k$ bits that approximates $f\oplus C$ well by our previous analysis. Consequently, $C\oplus g$ is a circuit of size $O(2^k/k + n)$ that approximates $f$ well, as desired.

\section{Preliminaries}
All logarithms are in base 2. $[n]\coloneqq \{1,\dots, n\}$. $\F_{2^n}$ denotes the finite field of $2^n$ elements, and each element will be identified by either a string in $\{0,1\}^{2^n}$ or integer in $[2^n]$ in the natural way. For a distribution $D$, $d\sim D$ is an element sampled from $D$. If $S$ is a set, we denote $s\sim S$ to be a uniformly random element from $S$. $\circ$ denotes string concatenation. We consider circuits with arbitrary gates of fan-in $2$ and arbitrary depth. A Boolean function $f:\{0,1\}^n\to \{0,1\}$ is a $k$-junta if $f(x) = g(x_S)$ for some subset $S\subset[n]$ of size $k$. For $x,y\in \{0,1\}^n$, we denote the distance between $x$ and $y$ to be the quantity $|\{i\in [n]: x_i\neq y_i\}|$. 

For $f,g:\{0,1\}^n\to \{0,1\}$, we say $f$ $\gamma$-correlates with $g$ over $H$ if $\Pr_{x\sim H}[f(x) = g(x)]\ge \frac{1}2 + \gamma$, and $f$ $\gamma$-approximates $g$ over $H$ if $\Pr_{x\sim H}[f(x) = g(x)]\ge \gamma$. If no $H$ is specified, it is assumed to be the uniform distribution over $\{0,1\}^n$.

We now define $k$-wise uniformity. 

\begin{definition}[$k$-wise uniformity \cite{HatamiHoza2024}]
    A distribution $D$ over $\{0,1\}^n$ is a \emph{$k$-wise uniform distribution} if, for all subsets $T\subset[n]$ of size $k$, the marginal distribution $(x_T)_{x\sim D}$ is uniform over $\{0,1\}^k$. A function $G: S\to\{0,1\}^n$ is a \emph{$k$-wise uniform generator} if $(G(s))_{s\sim S}$ is a $k$-wise uniform distribution.
\end{definition}

A crucial property of $4$-wise uniform strings is that they enjoy square-root anticoncentration with constant probability, just like a fully random string.

\begin{theorem}[{\cite[Corollary 3.1]{berger}}]
\label{thm:fourth-moment}
    Let $X$ be a $4$-wise uniform distribution over $\{-1,1\}^n$, and let $v\in \R^n$. We have \[\Pr_{x\sim X} \left[\left|\sum_{i=1}^n v_i x_i\right| \ge \sqrt{\frac{\sum_{i=1}^n v_i^2}{3}}\right]\ge \frac{2}{11}.\]
\end{theorem}

We will want $4$-wise generators such that, for a fixed seed, each output bit can be locally computed in small circuit size. The standard construction of $4$-wise uniform generators serves this purpose for us.

\begin{theorem}
\label{thm:4-wisequadratic}
   There exists a $4$-wise uniform generator $G:S\to \{0,1\}^{2^n}$ such that, for each $s\in S$ and $x\in [2^n]$, there exists a circuit $C_s$ of size $O(n^2)$ with $C_s(x) = G(s)_x$. 
\end{theorem}

\begin{proof}
    Define $\iota: \F_{2^n}\to \{0,1\}$ to map $x\in \F_{2^n}$ to the first bit of the binary encoding of $x$. Let $G: \F_{2^n}^4\to \F_2^{2^n}$ be defined by the evaluation map \[G(s) \coloneqq \left(\iota\left(\sum_{i=1}^4 s_i x^{i-1}\right)\right)_{x\in \F_{2^n}}.\] This is a $4$-wise uniform generator (see \cite[Theorem 2.2]{HatamiHoza2024}). Notice that, as a function of $x$, $G(s)_x$ is an evaluation of a degree-$3$ polynomial, which can be done in a circuit of size $O(n^2)$ by grade-school multiplication (better multiplication algorithms are known, but this suffices).
\end{proof}

% (omitted theorem comment left as-is)

This theorem suffices to prove \Cref{thm:random-tight}. A technical contribution of \cite{acr95}, and a key ingredient behind \Cref{thm:approx}, is a $4$-wise uniform generator that can be locally computed in linear circuit size, which is the best one could hope for.  

\begin{theorem}[\cite{acr95}]
\label{thm:4-wiselinear}
   There exists a $4$-wise uniform generator $G: S\to \{0,1\}^{2^n}$ such that, for each $s\in S$ and $x\in [2^n]$, there exists a circuit $C_s$ of size $O(n)$ with $C_s(x) = G(s)_x$. 
\end{theorem}

We give a self-contained proof of this in the appendix.

\section{Tightness of Impagliazzo's Hardcore Lemma}
\label{sec:imp}

We will now prove a lemma that shows how to construct small circuit approximators for arbitrary functions using $4$-wise uniform generators.
\begin{lemma}
\label{lem:approx-dist}
    Let $f:\{0,1\}^n$ be an arbitrary Boolean function, let $H$ be any distribution over $\{0,1\}^n$, and let $\gamma\in(\tfrac{27}{2^{n/2}},\tfrac{1}2)$. Let $G:\{0,1\}^m\to \{0,1\}^{2^n}$ be a $4$-wise uniform generator such that, for each $s\in \{0,1\}^m$, there exists a circuit $C_s$ of size $r$ with $C_s(x) = G(s)_x$. Then there exists a circuit $C$ of size $O\left(\frac{\gamma^2 2^n }{\log(\gamma^2 2^n )} + r\right)$ such that  \[\Pr_{x\sim H}[C(x) = f(x)]\ge \frac{1}2+\gamma.\]
\end{lemma}
\begin{proof}
    Denote $\ell \coloneqq \lfloor \log(1/363\gamma^2)\rfloor$, and note $n-\ell \ge 1$. Let $H'$ be the distribution over $\{0,1\}^{n-\ell}$ defined by the probability mass function \[H'(c) \coloneqq \Pr_{x\sim H} [x\in c\times \{0,1\}^{\ell}].\] This is the induced distribution of $H$ on the subcubes defined by the first $n-\ell$ bits of the input. For each $c$, define the conditional distribution over the subcube $c\times \{0,1\}^{\ell}$ by the function \[H_c(y) \coloneqq \Pr_{x\sim H}[x = c\circ y \mid x\in c\times \{0,1\}^{\ell}] = \frac{\Pr_{x\sim H}[x = c\circ y]}{H'(c)}.\] Let $G$ be the $4$-wise uniform generator guaranteed by the hypothesis. For a subcube $c$, denote the indicator variable \[I_c(s)\coloneqq \mathbf{1}\left(\left|\sum_{y\in \{0,1\}^{\ell}} H_c(y)(-1)^{f(c\circ y)+G(s)_{c\circ y}}\right| \ge \sqrt{\frac{\sum_{y\in\{0,1\}^\ell} H_c(y)^2 }{3}}\right).\] By \Cref{thm:fourth-moment}, we have that for each $c$ and random $s$, $\Pr_{s\sim \{0,1\}^m}[I_c(s) = 1]\ge \frac{2}{11}.$ Hence, \[\E_{s\sim\{0,1\}^m}\left[\E_{c\sim H'}[I_c(s)]\right] = \E_s\left[\sum_{c} H'(c)I_c(s)\right]\ge \frac{2}{11}\sum_c H'(c) = \frac{2}{11},\] and by an averaging argument there exists a choice of $s$ such that $\Pr_{c\sim H'}[I_c(s)=1]\ge 2/11$. Fix such an $s$. Now for any $c$ with $I_c(s) = 1$, we can use Cauchy-Schwarz to lower bound \[\left|\sum_{y\in \{0,1\}^{\ell}} H_c(y)(-1)^{f(c\circ y)+G(s)_{c\circ y}}\right| \ge \sqrt{\frac{\sum_{y\in\{0,1\}^\ell} H_c(y)^2 }{3}}\ge  \sqrt{\frac{(1/2^\ell)\sum_{y\in \{0,1\}^\ell} H_c(y)}{3}} = \sqrt{\frac{1}{2^\ell \cdot 3}}\] Now define $h:\{0,1\}^{n-\ell}\to\{0,1\}$ by the map \[h(c) = \mathbf{1}\left(\sum_{y\sim \{0,1\}^{\ell}} H_c(y)(-1)^{f(x)+G(s)_{c\circ y}} < 0\right),\] which encodes whether the subcube is positively or negatively correlated with the $4$-wise uniform string. We now set our approximator $g:\{0,1\}^n\to\{0,1\}$ to be $g(c\circ y)\coloneqq h(c)\oplus G(s)_{c\circ y}.$
             
    We can now write the correlation between $f$ and $g$ (with respect to $H$) as
    \begin{align*}
        \E_{x\sim H}[(-1)^{f(x)+g(x)}] 
        &= \E_{c\sim H'}\left[\sum_{y\in \{0,1\}^{\ell}} H_c(y)(-1)^{f(c\circ y)+g(c\circ y)}\right] \\ 
        &= \E_{c\sim H'}\left[(-1)^{h(c)}\sum_{y\in \{0,1\}^{\ell}} H_c(y)(-1)^{f(c\circ y) + G(s)_{c\circ y}}\right]\\ 
        &= \E_{c\sim H'}\left[\left|\sum_{y\in \{0,1\}^{\ell}} H_c(y)(-1)^{f(c\circ y) + G(s)_y}\right|\right] \\ 
        & \ge \Pr_{c\sim H'}[I_c(s)]\cdot \E_{c\sim H'}\left[\left|\sum_{y\in \{0,1\}^{\ell}} H_c(y)(-1)^{f(c\circ y) + G(s)_y}\right|: I_c(s) = 1\right]   
        \\
       & \ge \frac{2}{11}\cdot \sqrt{\frac{1}{2^\ell \cdot 3}} \ge 2\gamma.
    \end{align*}
     Hence, \[\Pr_{x\sim H}[f(x) = g(x)] = \frac{ \E_{x\sim H}[(-1)^{f(x)+g(x)}]+1}2 \ge \frac{1}2 + \gamma.\]

    Since $h$ depends only on the first $n-\ell$ bits, it can be constructed in circuit size $O\left(\frac{2^{n-\ell}}{n-\ell}\right) = O\left(\frac{\gamma^2 2^n}{\log(\gamma^2 2^n)}\right)$. By construction, $G(s)_y$ can be computed by some size-$r$ circuit $C_s$. Therefore, $g$ $\gamma$-correlates with $f$ and can be computed by a circuit of size $O\left(\frac{\gamma^2 2^n}{\log(\gamma^2 2^n)} + r\right) $, as desired.
\end{proof}

We are now ready to prove \Cref{thm:random-tight}.

\begin{proof}[Proof of \Cref{thm:random-tight}]
    We will first prove the first item. Pick a function $g:\{0,1\}^k\to \{0,1\}$ uniformly at random, and define $f(x)\coloneqq g(x_{\le k})$ to be $g$ on the first $k$ bits of the input. By a Chernoff bound, we know that for a fixed circuit $C:\{0,1\}^k\to\{0,1\}$ of size $\le s$ and $\delta < 0.49$, \[\Pr_f\left[\Pr_x[g(x) = C(x)]\ge 1-\delta\right]\le 2^{-\Omega(2^k)}.\] Set $2^k = \Omega(s\log s)$. Taking a union bound over all circuits of size $s$ (of which there are $s^{O(s)}$) we note that, with probability at most $s^{O(s)} 2^{-\Omega(2^k)} < 1$, $g$ can be $(1-\delta)$-approximated by a size-$s$ circuit. Fix a $g$ that cannot. We now show $f$ cannot be approximated by any circuit $C:\{0,1\}^n\to\{0,1\}$ of size $s$. Restricting the last $n-k$ bits of $C$ will result in a circuit on the first $k$ bits of size at most $s$. Hence, we can deduce \[\Pr_x[f(x) = C(x)] = \E_{x_{>k}}\left[\Pr_{x_{\le k}}[g(x_{ \le k }) = C(x_{\le k},x_{>k})]\right] < 1-\delta\] as desired.

    We will now prove this $f$ satisfies the second item. Consider an arbitrary distribution $H$ over $\{0,1\}^n$. The marginal distribution of $H$ on the first $k$ bits will be some distribution $H'$ over $\{0,1\}^k$. As $\gamma > s^{-1/2} > 27\cdot 2^{-k/2}$ for large enough $n$, \Cref{lem:approx-dist} and \Cref{thm:4-wisequadratic} give a circuit $C:\{0,1\}^k\to \{0,1\}$ of size $O\left(\frac{\gamma^2 2^k}{\log(\gamma^2 2^k)} + k^2\right)$ such that $\Pr_{x\sim H'}[g(x) = C(x)]\ge \frac{1}2 + \gamma.$ Letting $C':\{0,1\}^n\to \{0,1\}$ be the circuit that applies $C$ to the first $k$ bits, we note that \[\Pr_{x\sim H}[f(x) = C'(x)] = \Pr_{x\sim H'}[g(x) = C(x)] \ge \frac{1}2+\gamma.\] As $C'$ has the same size as $C$, and $\frac{\gamma^2 2^k}{\log(\gamma^2 2^k)} + k^2 = O\left(s\gamma^2\cdot \left(\frac{\log s}{\log(\gamma^2 s)}\right) + \log^2 s\right) = O_\eps(s \gamma^2)$ when $s\ge 1/\gamma^{2+\eps}$, $C'$ is a size-$O_\eps(s \gamma^2)$ circuit that $\gamma$-correlates with $f$ over $H$. Consequently, $f$ is a $k$-junta that cannot be $(1 - \delta)$-approximated by a size-$s$ circuit, but, over any $H$, $\gamma$-correlates with a size-$O_\eps(s \gamma^2)$ circuit, as desired.
\end{proof}

\begin{remark}
\label{rem:explicit-circuit-size}
   If \Cref{lem:approx-dist} and \Cref{thm:4-wiselinear} is used, rather than \Cref{lem:approx-dist} and \Cref{thm:4-wisequadratic}, we instead get $\gamma$-correlating circuits of size  $O\left(s\gamma^2\cdot \left(\frac{\log s}{\log(\gamma^2 s)}\right) + \log s\right) = O\left(s\gamma^2 \cdot \frac{\log s}{\log(s\gamma^2)}\right)$ for all $s = \Omega(1/\gamma^2)$.
\end{remark}

\begin{remark} Upon seeing the above proof, one might notice that it suffices to prove that a \emph{random} function has small approximating circuits over arbitrary $H$, rather than an \emph{arbitrary} function. \Cref{subsubsec:overview} gives a very simple argument to efficiently approximate a random function, so one might wonder why \emph{arbitrary} functions are considered. The arbitrariness of $H$ seems to force us to consider arbitrary $f$ (but this is not a rigorous claim). For example, can construct $H$ that renders the approximating circuit for the random function discussed in \Cref{subsubsec:overview} useless.   \end{remark}

\section*{Acknowledgements}

We thank Guy Blanc, Geoffrey Mon and David Zuckerman for illuminating discussions. We also thank Jeffrey Champion, Sabee Grewal, Rocco Servedio, and anonymous reviewers for comments on a draft of this work.

\bibliographystyle{alphaurl} \bibliography{references}

@inproceedings{nonexplicit-linear,
    title={On the Complexity of Coding},
    author={Gelfand, S. I. and Dobrushin, R. L. and Pinsker, M. S.},
    booktitle={Second International Symposium on Information Theory},
    pages={177--184},
    year={1973}
}

@article{lupanov,
title={Complexity of formula realization of functions of logical algebra},
volume={36}, 
DOI={10.2307/2270028}, 
number={3}, 
journal={Journal of Symbolic Logic}, 
author={Lupanov, O. B}, 
year={1971}, 
pages={547–548}}

@inproceedings{BarakHardcore2009,
author = {Boaz Barak and Moritz Hardt and Satyen Kale},
title = {The Uniform Hardcore Lemma via Approximate Bregman Projections},
booktitle = {{SODA}},
year = {2009},
doi = {10.1137/1.9781611973068.129}
}

@misc{Trevisan2009Blog,
  author       = {Luca Trevisan},
  title        = {Approximating a Boolean Function via Small Circuits},
  year         = {2009},
  howpublished = {\url{https://lucatrevisan.wordpress.com/2009/11/06/approximating-a-boolean-function-via-small-circuits/}},
  note         = {Accessed: 2025-08-08},
}

@inproceedings{Trevisan2003,
author = {Trevisan, Luca},
title = {On uniform amplification of hardness in NP},
year = {2005},
doi = {10.1145/1060590.1060595},
booktitle = {{STOC}},
}

@inproceedings{CasacubertaDworkVadhan2024,
author = {Casacuberta, S\'{\i}lvia and Dwork, Cynthia and Vadhan, Salil},
title = {Complexity-Theoretic Implications of Multicalibration},
year = {2024},
doi = {10.1145/3618260.3649748},
booktitle = {STOC}
}

@inproceedings{ReingoldTrevisanTulsianiVadhan2008,
  author    = {Omer Reingold and Luca Trevisan and Madhur Tulsiani and Salil Vadhan},
  title     = {Dense Subsets of Pseudorandom Sets},
  booktitle = {FOCS},
  year      = {2008},
  doi       = {10.1109/FOCS.2008.83}
}

@inproceedings{Holenstein2005,
  author    = {Thomas Holenstein},
  title     = {Key Agreement from Weak Bit Agreement},
  booktitle = {STOC},
  year      = {2005},
  doi       = {10.1145/1060590.1060688}
}

@inproceedings{ChenLyu2021,
  author    = {Lijie Chen and Xin Lyu},
  title     = {Inverse-Exponential Correlation Bounds and Extremely Rigid Matrices from a New Derandomized XOR Lemma},
  booktitle = {STOC},
  year      = {2021},
  doi       = {10.1145/3406325.3451132}
}

@inproceedings{SudanTrevisanVadhan2001,
author = {Sudan, Madhu and Trevisan, Luca and Vadhan, Salil},
title = {Pseudorandom generators without the XOR Lemma (extended abstract)},
year = {1999},
doi = {10.1145/301250.301397},
booktitle = {STOC},
}

@inproceedings{ODonnell-STOC-2002,
  author    = {Ryan O'Donnell},
  title     = {Hardness amplification within {NP}},
  booktitle = {STOC},
  year      = {2002},
  doi       = {10.1145/509907.510015}
}

@inproceedings {KS03,
author = { Klivans, Adam R. and Servedio, Rocco A. },
booktitle = {FOCS},
title = { Boosting and Hard-Core Sets },
year = {1999},
doi = {10.1109/SFFCS.1999.814638},
}

@inproceedings{Impagliazzo1995,
  author    = {Russell Impagliazzo},
  title     = {Hard-core Distributions for Somewhat Hard Problems},
  booktitle = {{FOCS}},
  year      = {1995},
  doi       = {10.1109/SFCS.1995.492584},
}

@article{LuTsaiWu2011,
  author       = {Chi-Jen Lu and Shi-Chun Tsai and Hsin-Lung Wu},
  title        = {Complexity of Hard-Core Set Proofs},
  journal      = {Computational Complexity},
  doi          = {https://doi.org/10.1007/s00037-011-0003-7},
  year         = {2011},
}

@article{HatamiHoza2024,
  author       = {Pooya Hatami and William M. Hoza},
  title        = {Paradigms for Unconditional Pseudorandom Generators},
  journal      = {Foundations and Trends® in Theoretical Computer Science},
  volume       = {16},
  year         = {2024},
  doi          = {10.1561/0400000109},
}

@article{berger,
author = {Berger, Bonnie},
title = {The Fourth Moment Method},
journal = {SIAM Journal on Computing},
year = {1997},
doi = {10.1137/S0097539792240005},

}

@article{RabaniShpilka2010,
  author    = {Yuval Rabani and Amir Shpilka},
  title     = {Explicit construction of a small {$\epsilon$}-net for linear threshold functions},
  journal   = {SIAM Journal on Computing},
  year      = {2010},
  doi       = {10.1137/090764190},
}

@article{Spielman1996,
  author    = {Daniel A. Spielman},
  title     = {Linear-Time Encodable and Decodable Error-Correcting Codes},
  journal   = {IEEE Transactions on Information Theory},
  year      = {1996},
  doi       = {10.1109/18.556668}
}

@article{acr95,
	
	author = {Alexander E. Andreev and Andrea E.F. Clementi and Jos{\'e} D.P. Rolim},
	doi = {https://doi.org/10.1016/S0304-3975(96)00217-4},
	journal = {Theoretical Computer Science},
	title = {Optimal bounds for the approximation of boolean functions and some applications},
	volume = {180},
	year = {1997},
}

@INPROCEEDINGS{BHKT24,
  author={Blanc, Guy and Hayderi, Alexandre and Koch, Caleb and Tan, Li-Yang},
  booktitle={{FOCS}}, 
  title={The Sample Complexity of Smooth Boosting and the Tightness of the Hardcore Theorem}, 
  year={2024},
  doi={10.1109/FOCS61266.2024.00092}}
\appendix
\section{Optimal Approximating Circuits: Proofs of \Cref{thm:4-wiselinear} and \Cref{thm:approx}}
In this section, we give motivation and a proof of the fundamental result of Andreev, Clementi, and Rolim \cite{acr95}. The exposition here has nontrivial overlap with that of Trevisan \cite{Trevisan2009Blog}.
\subsection{Motivation}
A classical probabilistic argument of Shannon says that there exists functions mapping $n$ bits to one bit that have circuit complexity $\Omega(2^n/n)$. What is remarkable is that this argument is \emph{tight}; Lupanov's theorem states that \emph{any} function $f:\{0,1\}^n\to\{0,1\}$ can be computed in circuit size $O(2^n/n)$ \cite{lupanov}.

Upon studying Shannon's argument, it is not hard to see that this lower-bound argument extends to circuits that only \emph{approximate} rather than \emph{compute}. Let $f:\{0,1\}^n\to \{0,1\}$ be a uniformly random function, and let $\gamma\in(0,\tfrac{1}2)$ be such that $\gamma^2 2^n\ge 2$. For a fixed circuit $C$, we have by the Chernoff bound that \[\Pr_f\left[\Pr_x[C(x) = f(x)] \ge \frac{1}2 + \gamma\right] \le 2^{-\Omega(\gamma^2 2^n)}.\] Setting $s = \Theta\left(\tfrac{\gamma^2 2^n}{\log(\gamma^2 2^n)}\right)$ and taking a union bound over all $s^{O(s)}$ circuits $C$ of size $s$ tells us that the probability $f$ is $\gamma$-correlated with a size-$s$ circuit is at most $s^{O(s)} 2^{-\Omega(\gamma^2 2^n)} < 0.1$.

It is now natural to ask whether this is tight: for any Boolean function, is there a size-$O\left(\tfrac{\gamma^2 2^n}{\log(\gamma^2 2^n)}\right)$ circuit that $\gamma$-approximates $f$? This is not true by considering the parity function. Any function that does not depend on all $n$ bits will be correct on the parity function on exactly half the inputs. Hence, an approximating circuit for parity must depend on all input bits and thus have circuit size at least $n$. We can then update our hypothesis and ask whether, for any Boolean function, there is a size-$ O\left( \tfrac{\gamma^2 2^n}{\log(\gamma^2 2^n)} + n\right)$ circuit that $\gamma$-correlates with it. 

A natural first approach is to try to use Lupanov's Theorem. In particular, we can construct a circuit that exactly computes $f$ inside a subcube of volume $2\gamma 2^n$, and then otherwise outputs a fixed bit, whichever matches $f$ better. This is guaranteed to exactly match $f$ within the subcube and match on at least half the inputs outside the subcube, giving a $\gamma$-correlating circuit. However, this will be a circuit of size $O\left(\frac{\gamma 2^n}{\log (\gamma 2^n)}\right)$.

Andreev, Clementi, and Rolim give a circuit construction matching the probabilistic bound, effectively establishing the approximation analog of Lupanov's theorem. We believe this result to be fundamental, but unfortunately, it appears to be relatively unknown to the community. Hence, we give a modern and simplified presentation of a slightly stronger result here (\Cref{thm:approx}).  

Notice that \Cref{lem:approx-dist} instantiated with \Cref{thm:4-wiselinear} gives \Cref{thm:approx}. \Cref{lem:approx-dist} was proven in \Cref{sec:imp}, so we now focus on the proof of \Cref{thm:4-wiselinear}. This was implicitly proven in \cite{acr95}, but we give a shorter proof using asymptotically good codes encodable in linear time \cite{nonexplicit-linear,Spielman1996}.

\subsection{Overview of \Cref{thm:4-wiselinear}}
The starting point of the construction is to consider the simpler task of a $2$-wise uniform generator. There is a classic 2-wise uniform generator mapping a nonzero seed of length $n$ to a string of size $2^n-1$, which is to simply output all nonempty $\F_2$-linear combinations of the seed; that is, $G(s)\coloneqq (\langle s, r\rangle)_{r\in \F_2^{n}\setminus\{0\}}$. Notice that, for a fixed $s$, the output of the $r$th bit as a function of $r$ is simply some parity of a subset of bits in $r$, which is trivially a circuit of size $O(n)$, as desired.

Why is this generator not a $4$-wise uniform generator? It is because of linear dependence. In particular, for nonzero vectors $x,y$, we have $G(s)_x + G(s)_y = G(s)_{x+y}.$ Hence, we do not even have $3$-wise uniformity, as the bits in indices $x,y$, and $x+y$ are correlated. This motivates the following idea: what if we only focused on a subset of indices $Y\subset\{0,1\}^n$ such that all distinct $x_1,x_2,x_3,x_4\in Y$ are linearly independent? Can we show that $(G(s)_y)_{y\in Y}$ is a $4$-wise independent string? Yes.

\begin{lemma}
\label{lem:4wise}
    Let $Y\subset \F_2^n$ be a subset such that, for all subsets $X\subset Y$ of size $4$, $X$ is linearly independent. Then $G:\{0,1\}^n\to \{0,1\}^Y$ defined by $G(s)_y = \langle y, s\rangle$ is a $4$-wise uniform generator.
\end{lemma}

\begin{proof}
    Consider arbitrary $X\subset Y$ of size $4$. We will show that over a uniform $s$, the string $(\langle s, x\rangle)_{x\in X}$ is uniform. Notice this string is simply $M\cdot s$, where $M$ is an $\F_2^{4\times n}$ matrix whose rows are the elements of $X$. Hence, every preimage of this map has the same size, namely that of the kernel of $M$. It remains to show that the image of $M$ is $\F_2^4$. But this is clear, as $X$ consists of linearly independent vectors, implying that the rank of $M$ is $4$.
\end{proof}

In light of this, we will try to construct a linear-size circuit $h: \{0,1\}^n\to\{0,1\}^{16n}$ such that, for all distinct $x_1,\dots, x_4$, the vectors $h(x_1), h(x_2), h(x_3), h(x_4)$ are linearly independent. Then our $4$-wise generator $G: \{0,1\}^{16n}\to \{0,1\}^{2^n}$ would be $G(s)_x =\langle h(x), s\rangle$, which will be a linear-size circuit. How do we construct such an $h$?

Perhaps a natural approach is to make $h$ randomly scatter the $x$'s randomly among $\{0,1\}^{16n}$. This actually works, because for a fixed $a\le 4$ and $x_1,\dots, x_a\in \F_2^n$, the probability $h(x_1)+\cdots + h(x_a) = 0$ is $2^{-16n}$. Taking a union bound over all tuples of size at most $4$ gets the desired result. Of course, the issue is that this is not a linear-size circuit: a random $h$ will have maximal circuit complexity. What if we let each bit of $h$ be a random function on only constantly many bits of $x$? 

More concretely, say we pick subsets $S_1,\dots, S_{16n}\subset [n]$ of constant size uniformly at random, and then let $g_i: \{0,1\}^{S_i}\to \{0,1\}$ be a random function. Define $h:\{0,1\}^n\to\{0,1\}^{16n}$ by $h(x) \coloneqq (g_1(x_{S_1}),\dots, g_{16n}(x_{S_{16n}}))$. This is of linear circuit size, and we are hoping the randomness of the $g_i$ keeps vectors linearly independent. Fix $x_1,\dots, x_a$ for $a\le 4$. We want the probability that $g_i((x_1)_{S_i}) + \dots + g_i((x_a)_{S_i}) = 0$ for all $i$ to be at most $1/\binom{2^n}{\le 4}$. Unfortunately, this need not be true. Say $x_1,\dots, x_4$ are within distance $2$ of each other. Then the probability that a random constant-sized $S_i$ satisfies $(x_1)_{S_i} = \cdots = (x_4)_{S_i}$ will be high. In this case, $g_i((x_j)_{S_i})$ will be guaranteed to all be the same for every $j\le 4$, and so $g_i((x_1)_{S_i})+\cdots + g_i((x_4)_{S_i}) = 0$. The key issue is that a randomly picked local view, $S_i$, might interpret $x_1,\dots, x_4$ as the same. This motivates the final trick of first encoding $x$ using an asymptotically good error-correcting code before picking our sets $S_i$. This will force different $x$'s to have very different encodings, and then a random set $S_i$ will indeed detect a difference. This will allow the randomness of $g_i$ to prevent linear dependencies from happening.

But are there asymptotically good codes encodable in linear circuit size? Indeed, non-explicit constructions of such codes were known to exist since 1974, thanks to Gelfand, Dobrushin, and Pinsker \cite{nonexplicit-linear}. Non-explicit constructions suffice for our application, but we mention that Spielman codes are explicit constructions of such a code \cite{Spielman1996}.

\begin{theorem}[\cite{nonexplicit-linear, Spielman1996}]
\label{thm:spielman}
    For any $n$ there exists a small enough constant $\delta, m\le 4n$, and an $O(n)$-sized circuit $C:\{0,1\}^n\to \{0,1\}^m$ such that for $x\neq y$, $C(x)$ and $C(y)$ have distance $\ge \delta m$.
\end{theorem}

With this primitive, we can construct our $4$-wise independence generator.

\subsection{Proofs of \Cref{thm:4-wiselinear} and \Cref{thm:approx}}

\begin{proof}[Proof of \Cref{thm:4-wiselinear}]
    Let $\enc:\{0,1\}^n\to \{0,1\}^{m}$ be an $O(n)$-sized circuit implied by \Cref{thm:spielman}.  We will show that there exist sets $S_1,\dots, S_{16n}\subset[m]$, each of size $\le \lceil 10/\delta\rceil$, and functions $g_1,\dots, g_{16n}$, with $g_i:\{0,1\}^{S_i}\to \{0,1\}$, such that \[G(s)_x \coloneqq \langle (g_1(\enc(x)_{S_1}),\dots,g_{16n}(\enc(x)_{S_{16n}})),s\rangle\] is a $4$-wise uniform generator. Once we have this, the desired result follows, as for a fixed $s$, $\enc(x)$ can be computed in linear circuit size, each $g_i$ can be computed in constant circuit size, and the parity of any subset of the $16n$ bits can be done in linear circuit size.
    
    By \Cref{lem:4wise}, it suffices to show the existence of $\{S_i\}, \{g_i\}$ such that for any $X\subset\F_2^n\setminus\{0\}$ of size $\le 4$, there exists $i\in[16n]$ such that $\sum_{x\in X} g_i(\enc(x)_{S_i}) \neq 0$. This will be done by the probabilistic method. In particular, we will pick each $S_i$ by selecting $\lceil 10/\delta \rceil$ elements uniformly and independently from $[m]$, and pick each $g_i:\{0,1\}^{S_i}\to \{0,1\}$ uniformly at random. Consider arbitrary $X\subset \{0,1\}^n\setminus \{0^n\}$ of cardinality at most $4$. 
    
    First assume $2\le |X|\le 4$. By construction of $\enc(\cdot)$, the strings $\{\enc(x)\}_{x\in X}$ will have pairwise distance $\ge \delta m$. Therefore, for a fixed $i\in[16n]$ and $x\neq x' \in X$, the probability that a random $S_i$ satisfies $\enc(x)_{S_i} = \enc(x')_{S_i}$ is at most $(1-\delta)^{|S_i|}$. Hence, the probability that $\enc(x)_{S_i} \neq \enc(x')_{S_i}$ for some $x\neq x'\in X$ is at least \[ 1-\binom{4}2(1-\delta)^{|S_i|} = 1-6e^{-10} \ge \frac{99}{100} \] by a union bound. Conditioned on this event, $\sum_{x\in X}g_i(\enc(x)_{S_i})$ is a uniform bit for a random $g_i$. Thus, for a fixed $i$, $\sum_{x\in X}g_i(\enc(x)_{S_i})\neq 0$ with probability at least $\frac{99}{100}\cdot \frac{1}2\ge\frac{1}3$. Since each coordinate is independent, the probability that $\sum_{x\in X}g_i(\enc(x)_{S_i}) = 0$ for all $i$ is at most $(2/3)^{16n}$.
    
    Now assume $X = \{x\}$. For any fixing of $\{S_i\}$ and for random $\{g_i\}$, $(g_i(\enc(x)_{S_i}))_{i\in [16n]}$ is a uniformly random string in $\{0,1\}^{16n}$, and is consequently $0$ with probability $ 2^{-16n} < (2/3)^{16n}$. 

Thus, by a union bound, all subsets $X$ of size at most 4 satisfy $\sum_{x\in X} g_i(\enc(x)_{S_i})\neq 0$ for some $i\in[16n]$ with probability at least \[1-\binom{2^n}{\le 4}(2/3)^{16n} \ge  1 - 2^{4n} (2/3)^{16n} > 0,\] implying the existence of such $\{S_i\}$ and $\{g_i\}$, and thereby yielding the result.
\end{proof}

With the help of \Cref{lem:approx-dist}, \Cref{thm:approx} is now immediate.

\begin{proof}[Proof of \Cref{thm:approx}]
    Simply use the $4$-wise uniform generator of \Cref{thm:4-wiselinear} in \Cref{lem:approx-dist}.
\end{proof}

\end{document}